\documentclass[letterpaper, 10 pt, conference]{ieeeconf}  

\IEEEoverridecommandlockouts                              

\usepackage{amsmath} 
\usepackage{amssymb}  
\usepackage{graphicx}
\usepackage{hyperref}
\usepackage{cite}
\usepackage{color}

\newtheorem{theorem}{Theorem}

\newtheorem{definition}{Definition}

\newtheorem{remark}{Remark}

\title{\LARGE \bf
Multiplier analysis of Lurye systems with power signals
}

\author{William P. Heath$^{1}$ and Joaquin Carrasco$^{1}$
\thanks{$^{1}$Control Systems and Robotics Group,
Department of Electrical and Electronic Engineering, Engineering Building A, University of Manchester, M13
9PL, UK. {\tt\small william.heath@manchester.ac.uk,
joaquin.carrascogomez@manchester.ac.uk}}
}

\begin{document}

\maketitle
\thispagestyle{empty}
\pagestyle{empty}

\begin{abstract}
Multipliers can be used to guarantee both the Lyapunov stability and input-output stability of Lurye systems with time-invariant memoryless slope-restricted nonlinearities. If a dynamic multiplier is used there is no guarantee the closed-loop system has finite incremental gain. It has been suggested in the literature that without this guarantee such a system may be critically sensitive to time-varying exogenous signals including noise. We show that multipliers guarantee the  power gain of the system to be bounded and quantifiable. Furthermore power may be measured about an appropriate steady state bias term, provided the multiplier does not require the nonlinearity to be odd. Hence dynamic multipliers can be used to guarantee Lurye systems have low sensitivity to noise, provided other exogenous systems have constant steady state. We illustrate the analysis 
with an example where the exogenous signal is a power signal with non-zero mean.
\end{abstract}

\section{INTRODUCTION}

We are concerned with  Lurye systems where the nonlinearity is memoryless, time-invariant and slope-restricted. The existence of a suitable OZF (O'Shea-Zames-Falb) multiplier can be used to guarantee finite-gain $\mathcal{L}_2$ stability \cite{Zames68,desoer75,Veenman16}. Modern tools can search for suitable multipliers and give an upper bound on the $\mathcal{L}_2$ gain \cite{Megretski04,Kao04,Turner12,Veenman16}.

It has been argued in the literature both specifically with respect to Lurye systems \cite{Kulkarni02,Waitman17} and more generally \cite{Zames66a,Fromion96,Angeli02,Sepulchre22,Chaffey23} that emphasis should be given to finite incremental gain. In particular if exogenous signals are not in $\mathcal{L}_2$ then desirable properties one might infer for linear systems do not necessarily carry over to nonlinear systems. Zames \cite{Zames66a} argues
 that a definition of closed-loop stability should require both continuity and boundedness, inter alia so that outputs are not ``critically sensitive to small changes in inputs --- changes such as those caused by noise''; Theorem 3 in \cite{Zames66a} gives conditions on the incremental positivity of the loop elements that are sufficient to achieve this.
Similarly it was known at the time \cite{Brockett66a} that the circle criterion could be used to guarantee ``the existence of unique steady-state oscillations (= absence of ``jump phenomena'' and subharmonics) in forced nonlinear feedback systems''. It was subsequently established \cite{Kulkarni02} that dynamic multipliers do not, in general, preserve the positivity of nonlinearities.
As noted in \cite{Kulkarni02} there is some irony that the definition of stability used by Zames and Falb in \cite{Zames68} does {\em not} require  closed-loop continuity.

Certainly lack of finite incremental gain can lead to undesirable effects.
A construction for a system where the input period is not preserved is given in \cite{Benes65} (although the specific example is closed-loop unstable). An example  of a Lurye system where small changes in input can lead to significant changes in output is given in \cite{Fromion04}.  We discuss a further example in this paper. However, we argue that finite incremental gain is not necessary to ensure insensitivity to noise signals. In particular the existence of a suitable OZF multiplier can be used to guarantee a Lurye system is insensitive to noise for a wide class of exogenous signal. This is timely in that OZF multipliers have recently been proposed for the design of control systems with saturation \cite{Veenman14,Bertolin22}.

 In this paper we characterise the response of such Lurye systems for two classes of exogenous signals that are not in $\mathcal{L}_2$. Specifically we consider power signals and signals with a constant bias. We observe that finite-gain stability ensures small power noise input leads to small power output when all other exogenous signals are in $\mathcal{L}_2$ (Theorem~\ref{thm:power_L2}). We define a notion of finite-gain offset stability (Definition~\ref{def:offset}) and observe similarly that finite-gain offset stability ensures small power noise input leads to small power output when all other exogenous signals are bias signals (Theorem~\ref{thm:power_offset}). Our main result (Theorem~\ref{thm:offset}) is to show that, provided we do not exploit any oddness of the nonlinearity, the existence of a suitable OZF multiplier guarantees such properties. Corresponding discrete-time results follow immediately and  we illustrate the results with a discrete-time example.

 \section{Preliminaries}

 \subsection{Signals and systems}

 Let $\mathcal{L}_2$ be the space of finite energy Lebesgue integrable signals on $[0,\infty)$  with norm
 \begin{equation}
 \|y\| = \left ( \int_0^{\infty}y(t)^2 \,dt\right )^{\frac{1}{2}}.
 \end{equation}
 Let  $\mathcal{L}_{2e}$ be the corresponding extended space (see for example  \cite{desoer75}). The truncation $y_T\in\mathcal{L}_2$ of $y\in\mathcal{L}_{2e}$ is given by
 \begin{equation}
     y_T(t) = \left \{ \begin{array}{lll}y(t) & \text{for} & 0\leq t \leq T,\\
     0 & \text{for} & T<t.\end{array}\right .
 \end{equation}
 \begin{definition}
 Let $\mathcal{P}\subset\mathcal{L}_{2e}$ be the space of finite power Lebesgue integrable signals on $[0,\infty)$ with seminorm
 \begin{equation}\label{def_P}
 \|y\|_P = \left (\limsup_{T\rightarrow\infty}\frac{1}{T}\|y_T\|^2\right )^{\frac{1}{2}}.
 \end{equation}
 We say $y$ is a power signal if $y\in\mathcal{P}$.
 Define the bias $\bar{y}\in\mathbb{R}$ of  a signal $y\in\mathcal{P}$ as
\begin{equation}
    \bar{y} = 
    \arg \min_{\bar{y}\in\mathbb{R}}\|(y-\bar{y}\theta)\|_P.
\end{equation}
We say $y$ is an $\mathcal{L}_2$-bias signal with bias $\bar{y}$ if $\bar{y}$ is unique and $y-\bar{y}\theta\in\mathcal{L}_2$.
\end{definition}

\begin{remark}
The limit superior in (\ref{def_P}) does not appear in standard definitions of power (e.g. \cite{vidyasagar93,Zhou96}) but is necessary to ensure $\mathcal{P}$ is a vector space \cite{Partington04,Mari96}\footnote{We are grateful to Andrey Kharitenko for this observation.}.
\end{remark}

 Let $\theta\in\mathcal{P}$ be the Heaviside step function given by 
     $\theta(t)=1$ for all  $t>0$.

 A map $H:\mathcal{L}_{2e}\rightarrow\mathcal{L}_{2e}$ is stable if $u\in\mathcal{L}_2$  implies $Hu\in\mathcal{L}_2$. It is finite-gain stable (FGS) if there is some $h<\infty$ such that $\|Hu\|\leq h \|u\|$ for all $u \in \mathcal{L}_2$. Its gain is the smallest such~$h$.

\begin{remark}\label{rem:ic1}
 Our definition of finite-gain stability carries the assumption that we have zero initial conditions.
 Non-zero initial conditions can be accommodated provided they can be represented with a nonlinear state-space description that is reachable and uniformly observable \cite{vidyasagar93}. 
 \end{remark}

\begin{definition}\label{def:offset}
Let $H:\mathcal{L}_{2e}\rightarrow\mathcal{L}_{2e}$. We say $H$ is offset stable if there is some function $H_0:\mathbb{R}\rightarrow\mathbb{R}$ such that if $u$ is an $\mathcal{L}_2$-bias signal with bias $\bar{u}$ then $Hu$ is an $\mathcal{L}_2$-bias signal with bias $H_0(\bar{u})$.
We call $H_0$ the steady state map of $H$.
Define $H_{\bar{u}}:\mathcal{L}_{2e}\rightarrow\mathcal{L}_{2e}$ as
\begin{equation}
H_{\bar{u}}u=H(u+\bar{u}\theta)-H_0(\bar{u}).
\end{equation}
It follows that $H$ is offset stable if $H_{\bar{u}}$ is stable for all $\bar{u}\in\mathbb{R}$.
We say $H$ is finite-gain offset stable (FGOS) if there is some $h<\infty$ such that $H_{\bar{u}}$ is FGS  with gain less than or equal to $h$ for all $\bar{u}\in\mathbb{R}$. We call the minimum such $h$ the offset gain of $H$.
\end{definition}

 \subsection{Lurye systems}\label{sec:Lurye}
 
We are concerned with the behaviour of the Lurye system (Fig.~\ref{fig:Lurye}) given by
\begin{equation}
y_1=\boldsymbol{G}u_1,\mbox{ } y_2=\boldsymbol{\phi} u_2,\mbox{ } u_1=r_1-y_2 \mbox{ and }u_2 = y_1+r_2.\label{eq:Lurye}
\end{equation}


The  Lurye system~(\ref{eq:Lurye}) is assumed to be well-posed with $\boldsymbol{G}:\mathcal{L}_{2e}\rightarrow\mathcal{L}_{2e}$  linear time invariant (LTI) causal and stable, and with $\boldsymbol{\phi}:\mathcal{L}_{2e}\rightarrow\mathcal{L}_{2e}$ memoryless and time-invariant. 
 We say such a memoryless, time-invariant $\boldsymbol{\phi}$ is characterised by $N:\mathbb{R}\rightarrow\mathbb{R}$ with $(\boldsymbol{\phi}(u))(t) = N(u(t))$.
 We will use ${G}$ to denote the transfer function corresponding to $\boldsymbol{G}$. Where appropriate we will consider either ${G}: j \mathbb{R}\rightarrow \mathbb{C}$ (i.e. ${G}(j\omega))$ or ${G}:\bar{\mathbb{C}}_{+}\rightarrow \mathbb{C}$ (i.e. ${G}(s)$) where $\bar{\mathbb{C}}_{+} = \{ s\in\mathbb{C}: Re(s)\geq 0\}$.

The nonlinearity $\boldsymbol{\phi}$, characterised by $N$, is assumed to be monotone in the sense that
$N(x_1)\geq N(x_2)$ for all $x_1\geq  x_2$. It is also assumed to be bounded in the sense that there exists a $C\geq 0$ such that $|N(x)|\leq C|x|$ for all $x\in\mathbb{R}$. We say  $\boldsymbol{\phi}$ is slope-restricted on $[0,k]$ if $0\leq (N(x_1) - N(x_2))/(x_1-x_2)\leq k$ for all $x_1\neq x_2$. We say $\boldsymbol{\phi}$ is odd if $N(x)=-N(-x)$ for all $x\in\mathbb{R}$.

 The Lurye system is stable if $r_1,r_2\in\mathcal{L}_2$  implies $u_1,u_2,y_1,y_2\in\mathcal{L}_2$. It is FGS if there is some $h<\infty$ such that
\begin{equation}
    \begin{split}
\|y_i\|\leq h (\|r_1\|+\|r_2\|)\text{ and }\|u_i\|\leq h (\|r_1\|+\|r_2\|),\\ 
\text{for }i=1,2\text{ and for all }r_1, r_2 \in \mathcal{L}_2.
    \end{split}
\end{equation}
\begin{remark}\label{rem:ic2}
Non-zero initial conditions can be accommodated in our definition of finite gain stability for such a Lurye system (c.f. Remark~\ref{rem:ic1}). Specifically, since the nonlinearity $\boldsymbol{\phi}$ is Lipschitz and the LTI transfer function $G$ admits a minimal state-space representation,  non-zero initial conditions can be accommodated by extending the time line backwards and including some fictitious exogenous signal over this extension. See \cite{vidyasagar93}, pp290-291.
\end{remark}
Since $\boldsymbol{G}$ is LTI stable we can set $r_1=0$ without loss of generality, and it is sufficient for finite-gain stability that there is some $h<\infty$ with
$\|y_2\|\leq h \|r_2\|$.
We will denote $H$ as the map from $r_2$ to $y_2$ and define the gain of the Lurye system to be the gain of $H$. Similarly for offset stability and offset gain.

\begin{figure}[htbp]
\begin{center}
\ifx\JPicScale\undefined\def\JPicScale{1}\fi
\unitlength \JPicScale mm
\begin{picture}(60,30)(0,0)
\linethickness{0.3mm}
\put(0,25){\line(1,0){7.5}}
\put(7.5,25){\vector(1,0){0.12}}
\linethickness{0.3mm}
\put(10,5){\line(1,0){15}}
\linethickness{0.3mm}
\put(10,5){\line(0,1){17.5}}
\put(10,22.5){\vector(0,1){0.12}}
\linethickness{0.3mm}
\put(12.5,25){\line(1,0){12.5}}
\put(25,25){\vector(1,0){0.12}}
\linethickness{0.3mm}
\put(35,25){\line(1,0){15}}
\linethickness{0.3mm}
\put(50,7.5){\line(0,1){17.5}}
\linethickness{0.3mm}
\put(35,5){\line(1,0){12.5}}
\put(35,5){\vector(-1,0){0.12}}
\put(0,30){\makebox(0,0)[cc]{$r_1$}}

\put(42.5,30){\makebox(0,0)[cc]{$y_1$}}

\linethickness{0.3mm}
\put(25,30){\line(1,0){10}}
\put(25,20){\line(0,1){10}}
\put(35,20){\line(0,1){10}}
\put(25,20){\line(1,0){10}}
\linethickness{0.3mm}
\put(25,10){\line(1,0){10}}
\put(25,0){\line(0,1){10}}
\put(35,0){\line(0,1){10}}
\put(25,0){\line(1,0){10}}
\put(30,5){\makebox(0,0)[cc]{$\boldsymbol\phi$}}

\put(30,25){\makebox(0,0)[cc]{$\boldsymbol{G}$}}

\linethickness{0.3mm}
\put(10,25){\circle{5}}

\put(12.5,20){\makebox(0,0)[cc]{$-$}}

\put(17.5,30){\makebox(0,0)[cc]{$u_1$}}

\linethickness{0.3mm}
\put(50,5){\circle{5}}

\linethickness{0.3mm}
\put(52.5,5){\line(1,0){7.5}}
\put(52.5,5){\vector(-1,0){0.12}}
\put(60,2.5){\makebox(0,0)[cc]{$r_2$}}

\put(42.5,2.5){\makebox(0,0)[cc]{$u_2$}}

\put(17.5,2.5){\makebox(0,0)[cc]{$y_2$}}

\end{picture}
\end{center}
\caption{Lurye system.} 
\label{fig:Lurye}
\end{figure}
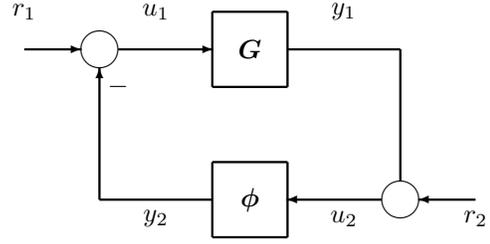


\subsection{Multiplier theory}

\begin{definition}[\cite{Zames68}]\label{def2a}
Let $\mathcal{M}$ be the class of 
transfer functions ${M}:j \mathbb{R}\rightarrow\mathbb{C}$  of systems whose (possibly non-causal) impulse response is given by
\begin{equation}\label{def_m}
m(t) = m_0 \delta(t)-h(t)-\sum_{i=1}^{\infty}h_i \delta(t-t_i),
\end{equation}
with
$ h(t)  \geq 0 \mbox{ for all } t\in\mathbb{R}\mbox{, }h_i\geq 0 \mbox{ for all } i$
and
\begin{equation}\label{OZF_ineq}
\| h\|_1+\sum_{i=1}^{\infty} |h_i| 
 \leq  m_0.
\end{equation}
We say ${M}$ is an OZF multiplier if ${M}\in\mathcal{M}$.
\end{definition}
\begin{definition}[\cite{Zames68}]\label{def2b}
Let $\mathcal{M}_{\text{odd}}$ be the class of 
transfer functions  ${M}:j \mathbb{R}\rightarrow\mathbb{C}$ of systems  whose (possibly non-causal) impulse response is given by (\ref{def_m})
with (\ref{OZF_ineq}).
We say ${M}$ is an OZF multiplier for odd nonlinearities if ${M}\in\mathcal{M}_{\text{odd}}$.
\end{definition}

\begin{definition}\label{def1}
Let 
${M}:j \mathbb{R}\rightarrow\mathbb{C}$ and let ${G}:j \mathbb{R}\rightarrow\mathbb{C}$. 
We say ${M}$ is suitable for ${G}$ if there exists $\varepsilon>0$ such that
\begin{align}
\mbox{Re}\left \{
				{M}(j\omega) {G}(j\omega)
			\right \} > \varepsilon\mbox{ for all } \omega \in \mathbb{R}.
\end{align}
\end{definition}

\begin{theorem}[\cite{Zames68,desoer75}]

If there is 
an ${M}\in\mathcal{M}$ ($\mathcal{M}_{\text{odd}}$) suitable for~${G}$ then the Lurye system (\ref{eq:Lurye}) is FGS for any memoryless time-invariant (odd) monotone bounded nonlinearity $\boldsymbol{\phi}$. 
Furthermore, if there is 
 an ${M}\in\mathcal{M}$  ($\mathcal{M}_{\text{odd}}$) suitable for $1/k+{G}$ then the Lurye system (\ref{eq:Lurye}) is FGS for any memoryless time-invariant (odd) slope-restricted nonlinearity $\boldsymbol{\phi}$ in $[0,k]$.
\end{theorem}

\begin{remark}\label{rem:strict}
    We have used a non-strict inequality in (\ref{OZF_ineq}). While the early literature uses a strict inequality \cite{OShea67,Zames68, desoer75} some modern literature (e.g. \cite{Megretski97}) uses a non-strict inequality while some (e.g. \cite{Veenman16}) retains the strict inequality. The distinction is discussed in \cite{Carrasco12}.
\end{remark}



 \subsection{Continuity}
 \begin{definition}\cite{Zames66a}
 The incremental gain of $H:\mathcal{L}_{2e}\rightarrow\mathcal{L}_{2e}$ is the supremum of 
 $\|(Hx)_T-(Hy)_T\| / \|x_T-y_T\|$
 over all $x,y\in\mathcal{L}_{2e}$ and all $T>0$ for which $\|x_T-y_T\|\neq 0$.
 \end{definition}
 Finite-gain stability does not guarantee finite incremental gain: if $H$ is FGS and $u_1,u_2\in\mathcal{L}_2$ the ratio $R=\|Hu_1-Hu_2\|/\|u_1-u_2\|$ may be arbitrarily large. In particular, suppose $v_1,v_2$ are power signals with $v_1-v_2\in\mathcal{L}_2$ but $v_1,v_2\notin\mathcal{L}_2$ and suppose $H$ is FGS but $Hv_1-Hv_2\notin\mathcal{L}_2$. Let $u_1,u_2$ be the truncations $u_1=(v_1)_T$ and $u_2=(v_2)_T$. Then $R\rightarrow\infty$ as $T\rightarrow\infty$.

 As noted in the Introduction, dynamic multipliers do not, in general, preserve the positivity of nonlinearities \cite{Kulkarni02}. An example of a Lurye system where multipliers guarantee finite gain stability but where Lipschitz continuity is lost is given in \cite{Fromion04}.

\section{Power analysis}

In this section we consider FGS systems. The application to Lurye systems where finite-gain stablilty is guaranteed by the existance of a suitable OZF multiplier is immediate.

Zames \cite{Zames66a} argues that ``in order to behave properly'' a system's ``outputs must not be critically sensitive to small changes in inputs - changes such as those caused by noise.'' Further he argues that the input-output map must be continuous to ensure it is ``not critically sensitive to noise.'' Here we show that if the noise is a power signal then finite-gain stability suffices.

The following is standard, at least for linear systems \cite{Zhou96}.

\begin{theorem}\label{thm:power}
Suppose $u\in\mathcal{P}$ and $y=Hu$ where $H$ is FGS with gain $h$.  Then $y\in\mathcal{P}$ with $\|y\|_P\leq h \|u\|_P$.
\end{theorem}
\begin{proof}We find
\begin{align}
    \|y\|_P^2 & = \limsup_{T\rightarrow \infty} \frac{1}{T}\|(Hu)_T\|^2, \nonumber \\
     & = \limsup_{T\rightarrow \infty} \frac{1}{T}\|(Hu_T)_T\|^2\text{ since $H$ is causal,}\nonumber \\
     & \leq \limsup_{T\rightarrow \infty} \frac{1}{T}\|Hu_T\|^2,\nonumber\\
     & \leq \limsup_{T\rightarrow \infty} \frac{1}{T} h^2\|u_T\|^2,\nonumber\\
     & = h^2 \|u\|_P^2. 
\end{align}
\end{proof}
Hence if we define the power gain of $H$ to be
\begin{equation}
    h_P = \sup_{u\in\mathcal{P},\|u\|_P>0} \frac{\|Hu\|_P}{\|u\|_P},
\end{equation}
then $h_P\leq h$.

Suppose $u_1,u_2\in\mathcal{P}$ and $H$ is FGS with gain $h$. Since $\|\cdot\|_P$ is a seminorm we have the triangle inequalities
\begin{align}
        \|Hu_1 & \pm  Hu_2\|_P^2 \nonumber\\
          & \leq \|Hu_1\|_P^2+\|Hu_2\|_P^2 + 2 \|Hu_1\|_P\|Hu_2\|_P\nonumber\\
         & \leq h^2 \left [\|u_1\|_P^2+\|u_2\|_P^2+2\|u_1\|_P\|u_2\|_P\right ],
\end{align}
and
\begin{align}\label{ineq:power}
        \|H(u_1+& u_2)\|_P^2 \leq h^2 \|u_1+u_2\|^2 \nonumber\\
        & \leq h^2 \left [\|u_1\|_P^2+\|u_2\|_P^2+2\|u_1\|_P\|u_2\|_P\right ].
\end{align}
In particular we may say:
\begin{theorem}\label{thm:power_L2}
Suppose $H:\mathcal{L}_{2e}\rightarrow\mathcal{L}_{2e}$ is FGS with gain $h$. Suppose further that $u_1\in\mathcal{L}_2$ and $u_2\in\mathcal{P}$.  Then
\begin{equation}
        \|H(u_1 + u_2)\|_P  \leq h \|u_2\|_P.
\end{equation} 
\end{theorem}
\begin{proof}
The result follows immediately from (\ref{ineq:power}) since 
 $\|u_1\|_P=0$.
\end{proof}
Hence, if we add small power noise to an $\mathcal{L}_2$ input signal then the output power is small provided the system is FGS, irrespective of the continuity of the input-output map.

\section{Bias signal analysis}
Theorem~\ref{thm:power_L2} has an immediate counterpart when $u_1$ is a bias signal and $H$ is FGOS.
\begin{theorem}\label{thm:power_offset}
Suppose $H:\mathcal{L}_{2e}\rightarrow\mathcal{L}_{2e}$ is FGOS with steady state map $H_0$ and offset gain $h$. Suppose further that $u_1$ is a bias signal with bias $\bar{u}_1$ and  $u_2\in\mathcal{P}$.  Then
\begin{equation}
 \|H(u_1+u_2)-H_0(\bar{u}_1)\|_P  \leq h\|u_2\|_P.
 \end{equation}
\end{theorem}
\begin{proof}
The result follows from Theorem~\ref{thm:power_L2} 
since
\begin{equation}
 \|H(u_1+u_2)-H_0(\bar{u}_1)\|_P = \|H_{\bar{u}}(u_1-\bar{u}_1\theta+u_2)\|_P,
 \end{equation}
 where $H_{\bar{u}}$ is given by Definition~\ref{def:offset}. 
\end{proof}
Hence, if we add small power noise to a input bias signal then the output power (measured about the noise-free output bias) is small provided the system is FGOS, irrespective of the continuity of the input-output map.

The application to a Lurye system is immediate provided we can show the system is FGOS. It turns out that the existence of a suitable OZF multiplier $M\in\mathcal{M}$ guarantees this, but the existence of a suitable $M\in\mathcal{M}_{\text{odd}}$ does not.

\begin{theorem}\label{thm:offset}
If there is 
an ${M}\in\mathcal{M}$  suitable for~${G}$ then the Lurye system (\ref{eq:Lurye}) is FGOS for any memoryless time-invariant monotone bounded nonlinearity $\boldsymbol{\phi}$. If there is 
 an ${M}\in\mathcal{M}$  suitable for $1/k+{G}$ then the Lurye system (\ref{eq:Lurye}) is FGOS for any memoryless time-invariant slope-restricted nonlinearity $\boldsymbol{\phi}$ in $[0,k]$.
\end{theorem}

\begin{proof}
Without loss of generality suppose $r_1=0$ and let $H$ be the map from $r_2$ to $y_2$. Let $r_2$ be a bias signal with bias $\bar{r}_2$. Suppose first that $r_2$ is constant 
$r_2=\bar{r}_2\theta$
and $G$ is a fixed gain $G=\bar{G}$. The monotonicity of $\boldsymbol{\phi}$ ensures there is a unique fixed solution 
$y_1=\bar{y}_1\theta$ \cite{desoer75}.
This defines our candidate bias function $H_o(\bar{r}_2)=\bar{y}_1$. 
The input to the nonlinearity is $u_2=\bar{u}_2\theta$ where $\bar{u}_2=\bar{r}_2+\bar{y}_1$ and its output is $y_2=\bar{y}_2\theta$ where $\bar{y}_2=N(\bar{u}_2)$ and the nonlinearity $\boldmath{\boldsymbol{\phi}}$ is characterised by $N:\mathbb{R}\rightarrow\mathbb{R}$. We can define a memoryless, time-invariant, bounded and monotone nonlinearity $\boldsymbol{\phi}_{\bar{r}_2}$ characterised by $N_{\bar{r}_2}$ where
\begin{equation}
N_{\bar{r}_2}(x) = N (x+\bar{u}_2) - \bar{y}_2 \text{ for all }x\in\mathbb{R}.
\end{equation}
This in turn defines a normalized Lurye system with linear element $G$ and nonlinear element $\boldsymbol{\phi}_{\bar{r}_2}$. Denote $H_{\bar{r}_2}$ as the map from $r_2-\bar{r}_2\theta$ to $y_1-\bar{y}_1\theta$. If $M\in\mathcal{M}$ is suitable for $G$ then $H_{\bar{r}_2}$ is FGS. It follows that $H$ is FGOS.

If $\boldsymbol{\phi}$ is in addition slope-restricted on $[0,k]$ then $\boldsymbol{\phi}_{\bar{r}_2}$ is also slope-restricted on $[0,k]$.
\end{proof}

\begin{remark}
There is no corresponding result when $M\in\mathcal{M}_{\text{odd}}-\mathcal{M}$. In this case  $\boldsymbol{\phi}_{\bar{r}_2}$ need not be odd even if $\boldsymbol{\phi}$ is odd. Hence if $M\in\mathcal{M}_{\text{odd}}-\mathcal{M}$ is suitable for $G$ there is no guarantee that $H_{\bar{r}_2}$ is stable.
\end{remark}

\begin{remark}
Suppose $M$ is suitable for $G$ (or for $1/k+G$) and can be used to ensure the $\mathcal{L}_2$ gain of $H$ is bounded above by $h$. It follows from the proof of Theorem~\ref{thm:offset} that the offset gain of $H$ is also bounded above by $h$.
\end{remark}

\section{Example}

\subsection{Discrete time preliminaries}

Although our development has been for continuous-time systems, we illustrate with a discrete-time example. Discrete-time counterparts to our results follow immediately, with 
the spaces $\ell$ of real-valued sequences $h:\mathbb{Z}^+\rightarrow \mathbb{R}$ and $\ell_2$ of square-summable sequences $h:\mathbb{Z}^+\rightarrow \mathbb{R}$  taking the places of $\mathcal{L}_{2e}$ and   $\mathcal{L}_2$ respectively. Similarly $\ell_2$ gain takes the place of $\mathcal{L}_2$ gain.

Specifically, let $\mathbb{D}$ denote the unit circle. Respective counterparts to Definitions~\ref{def2a},~\ref{def2b} and~\ref{def1} are as follows.

\begin{definition}\label{def2ad}
Let $\mathcal{M}^d$ be the class of 
discrete-time transfer functions ${M}:\mathbb{D}\rightarrow\mathbb{C}$  of systems whose (possibly non-causal) impulse response is given by
\begin{equation}\label{def_md}
m(t) = m_0 \delta(t)-\sum_{i\neq 0}h_i \delta(t-i),
\end{equation}
with
$ h_i\geq 0 \mbox{ for all } i$
and
\begin{equation}\label{OZF_ineqd}
 \sum_{i\neq 0} |h_i| 
 \leq  m_0.
\end{equation}
We say ${M}$ is a discrete-time OZF multiplier if ${M}\in\mathcal{M}^d$.
\end{definition}
\begin{definition}\label{def2bd}
Let $\mathcal{M}^d_{\text{odd}}$ be the class of 
discrete-time transfer functions  ${M}:\mathbb{D}\rightarrow\mathbb{C}$ of systems  whose (possibly non-causal) impulse response is given by (\ref{def_md})
with (\ref{OZF_ineqd}).
We say ${M}$ is a discrete-time OZF multiplier for odd nonlinearities if ${M}\in\mathcal{M}^d_{\text{odd}}$.
\end{definition}

\begin{definition}\label{def1d}
Let 
${M}:\mathbb{D}\rightarrow\mathbb{R}$ and let ${G}:\mathbb{D}\rightarrow\mathbb{C}$. 
We say ${M}$ is suitable for ${G}$ if
\begin{align}
\mbox{Re}\left \{
				{M}(e^{j\omega}) {G}(e^{j\omega})
			\right \} > 0\mbox{ for all } \omega \in [0,2\pi).
\end{align}
\end{definition}

The counterpart to Theorem 1 is direct and standard \cite{Willems68,Willems71}. Direct counterparts to Theorems 2-5 follow immediately.

\begin{remark}
Once again we have used a non-strict inequality in the definition of the OZF multipliers (\ref{OZF_ineqd}); see Remark~\ref{rem:strict}.
The original literature \cite{Willems68,Willems71} includes time-varying multipliers but this is unnecessary \cite{Kharitenko23}.
\end{remark}

\subsection{The example}

Consider the Lurye system (\ref{eq:Lurye}) where $\boldsymbol{\phi}$ is the saturation  characterised by $N:\mathbb{R}\rightarrow\mathbb{R}$ with
$N(x)=x$ when $|x|<1$ and
$N(x)=x/|x|$ when $|x|\geq 1$,
and $\boldsymbol{G}$ has the discrete-time transfer function
\begin{equation}\label{defG}
    G(z)=g\frac{2z+0.92}{z(z-0.5)},
\end{equation}
and where the gain $g$ takes one of three values: $g=0.6$, $g=0.8$ or $g=1$. NB we have considered this Lurye system previously with $g=1$ \cite{Heath15,Heath22}.

Following standard analysis (e.g. \cite{Veenman16}), we may say that if $M\in \mathcal{M}^d_{odd}$ is suitable for $G$ and 
\begin{align}
2 & \text{Re} [M(e^{j\omega})(1+G(e^{j\omega}))]\gamma^2   \nonumber\\
 & -(|G(e^{j\omega})|^2+|M(e^{j\omega})|^2)\gamma- 2 \text{Re}[M(e^{j\omega})] >0,
\end{align}
for all $\omega\in[0,2\pi)$ and for some $\gamma\geq 1$ then the $l_2$ gain from $r_2$ to $u_2$ is bounded above by $\gamma$. If $M\in\mathcal{M}^d\subset \mathcal{M}^d_{odd}$ then the offset gain is also bounded above by $\gamma$ (see Remark~4).

\begin{description}
\item{When $g=0.6$} finite incremental  gain of the Lurye system may be established via the circle criterion. Specifically $\text{Re}\left [1+G(e^{j\omega})\right ]>0$ for all $\omega\in[0,2\pi)$.  Although dynamic multipliers are not required to establish $\ell_2$ stability, they may still be useful to find reduced upper bounds on the $\ell_2$ gain, and hence on the offset gain. In this case the circle criterion alone establishes an upper bound $h\leq 16.3156$ but the multiplier $M(e^{j\omega})=1-0.66e^{-j\omega}$ establishes an upper bound  $h \leq 4.1795$.

\item{When $g=0.8$} the circle criterion can no longer be applied, but there are multipliers  $M\in\mathcal{M}^d$ suitable for $G$. It follows from the discrete-time counterpart of 
Theorem~\ref{thm:offset}
that the Lurye system is FGOS. For example the multiplier $M=1-0.85e^{-j\omega}$ is suitable for $G$ and establishes an upper bound $h\leq 12.8983$ on the $\ell_2$ gain, and hence on the offset gain.

\item{When $g=1$} we find
     $\angle [1+gG(e^{2\pi j/3}) ]  = -\pi+\text{atan} \frac{31\sqrt{3}}{48}$ \\
    $< -\frac{2\pi}{3}$.
It follows from the phase limitations at single frequencies \cite{Zhang22} that there is no $M\in\mathcal{M}^d$ suitable for $G$. Nevertheless there are multipliers $M\in\mathcal{M}^d_{odd}$ suitable for $G$. It follows that the Lurye system is FGS, but not necessarily FGOS in this case. For example the multiplier $M=1+0.9e^{j\omega}$ is suitable for $G$ and establishes an upper bound $h \leq 31.332$ on the $\ell_2$ gain.
\end{description}

\def\figwidth{0.41}

    \begin{figure}[t]
        \begin{center}
        \includegraphics[width=\figwidth\textwidth]{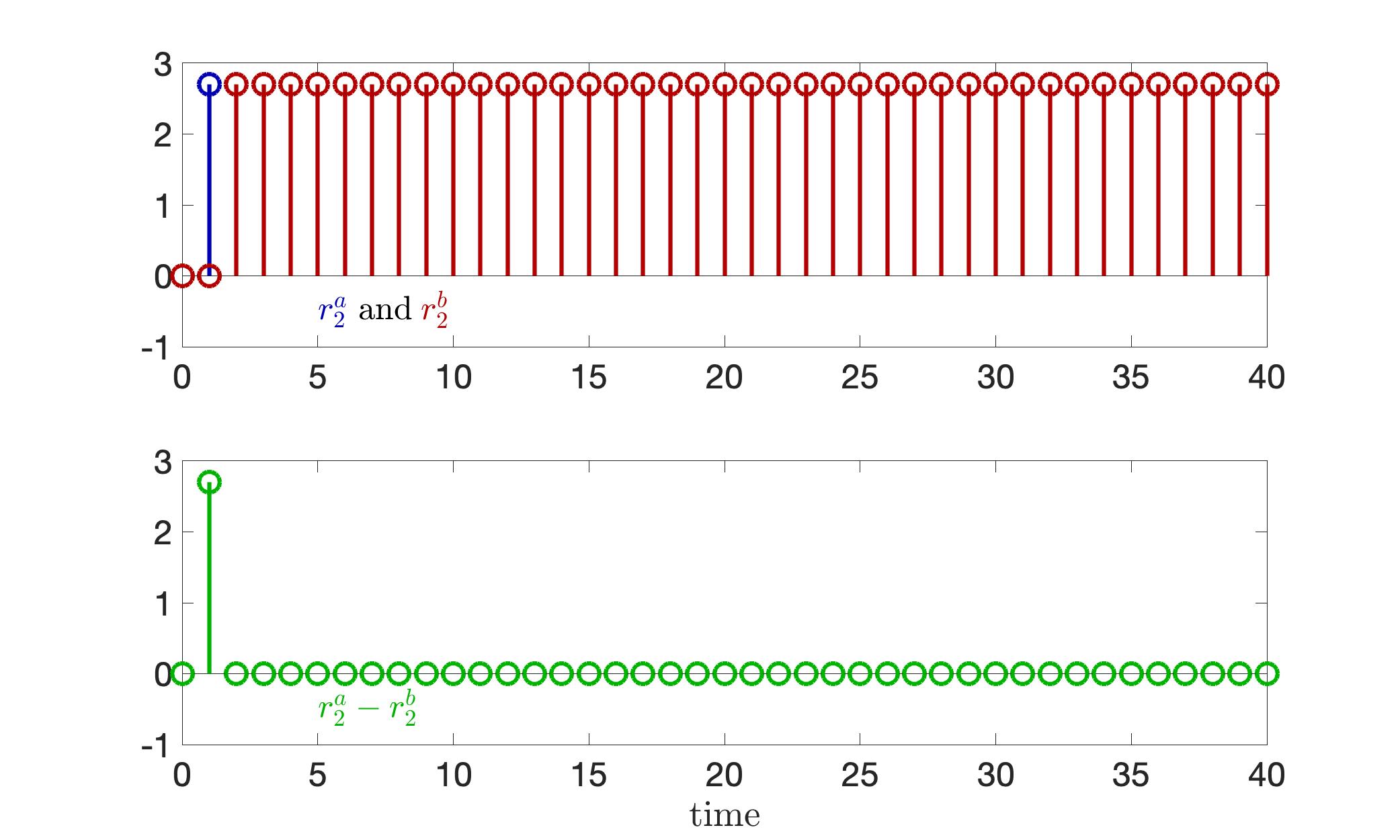}
        \end{center}
        \caption{Deterministic step signals $r_2^a$ and $r_2^b$. Both are bias signals, with $r_2^b$ a delayed version of $r_2^a$. Their difference $r_2^a-r_2^b\in\ell_2$.}\label{Fig:step1}
        \begin{center}
        \includegraphics[width=\figwidth\textwidth]{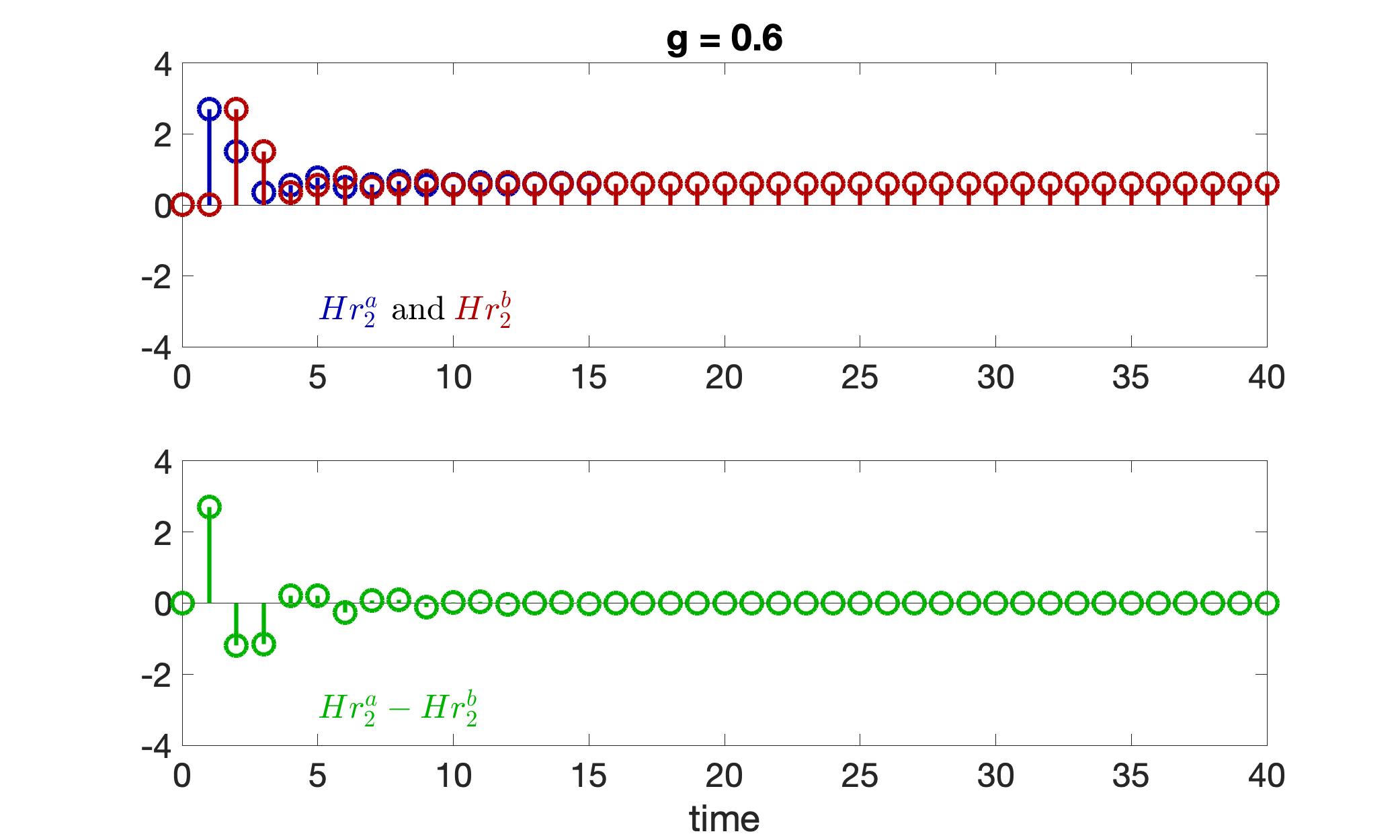}
        \end{center}
        \caption{Respective responses $Hr_2^a$ and $Hr_2^b$ to deterministic step signals $r_2^a$ and $r_2^b$ when the gain is $g=0.6$. The circle criterion guarantees finite incremental gain and hence the difference  $Hr_2^a-Hr_2^b\in\ell_2$.}\label{Fig:step2}
    \end{figure}

\subsection{Deterministic step response}
    Let $r_2$ be either the step response $r_2^a$ or the delayed step response $r_2^b$ given by
    \begin{align}
        r_2^a(t) & = \left \{\begin{array}{lll}0 & \text{when} & t=0,\\ 2.7 & \text{when} & t\geq 1,\end{array} \right . \nonumber\\
        r_2^b(t) & = \left \{\begin{array}{lll}0 & \text{when} & t=0,1,\\ 2.7 & \text{when} & t\geq 2.\end{array} \right .
    \end{align}
    Although $r_2^a\notin\ell_2$ and $r_2^b\notin\ell_2$ we have $r_2^a-r_2^b\in\ell_2$ (see Fig.~\ref{Fig:step1}).

    When $g=0.6$ (Fig.~\ref{Fig:step2}) the closed-loop system is incrementally stable so the difference $Hr_2^a-Hr_2^b\in\ell_2$. When $g=0.8$ (Fig.~\ref{Fig:step3}) the results of this paper show that the closed-loop system is FGOS so once again $Hr_2^a-Hr_2^b\in\ell_2$. By contrast, when $g=1$  (Fig.~\ref{Fig:step4}) the outputs $Hr_2^a$ and $Hr_2^b$ oscillate in steady state, and hence the difference $Hr_2^a-Hr_2^b$ also oscillates (since $Hr_2^b$ is a delayed version of $Hr_2^a$); i.e. the difference $Hr_2^a-Hr_2^b\notin\ell_2$ is a power signal. If we consider the response to the truncated signals $(r_2^a)_T,(r_2^b)_T,\in\ell_2$ then we find that when $g=1$ we obtain
    \begin{equation}
        \|H (r_2^a)_T-H(r_2^b)_T\| / \|(r_2^a)_T-(r_2^b)_T\|\rightarrow\infty \text{ as } T\rightarrow\infty.
    \end{equation}
    
    \begin{remark}
        Since $r_2$ is a power signal, Theorem~\ref{thm:power} can be applied to all three cases. However for the case $g=1$, since the closed-loop system is not FGOS, the power must be normalized around $0$ rather than around the steady state values.
    \end{remark}

    \begin{figure}[tbp]
        \begin{center}
        \includegraphics[width=\figwidth\textwidth]{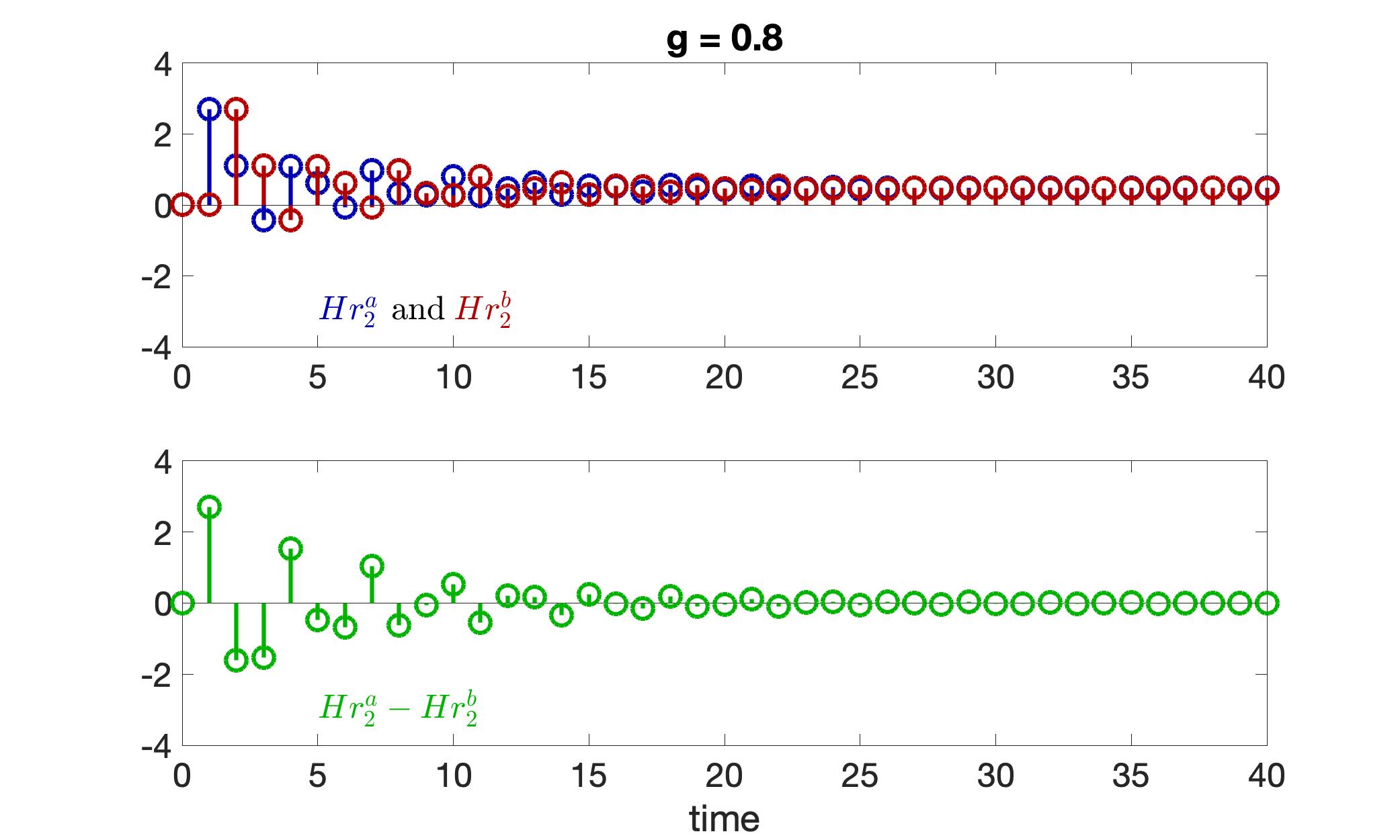}
        \end{center}
        \caption{Respective responses $Hr_2^a$ and $Hr_2^b$ to deterministic step signals $r_2^a$ and $r_2^b$ when the gain is $g=0.8$. There is a suitable multiplier $M\in\mathcal{M}^d$ for $G$. It follows that $H$ is FGOS and the difference  $Hr_2^a-Hr_2^b\in\ell_2$.}\label{Fig:step3}
        \begin{center}
        \includegraphics[width=\figwidth\textwidth]{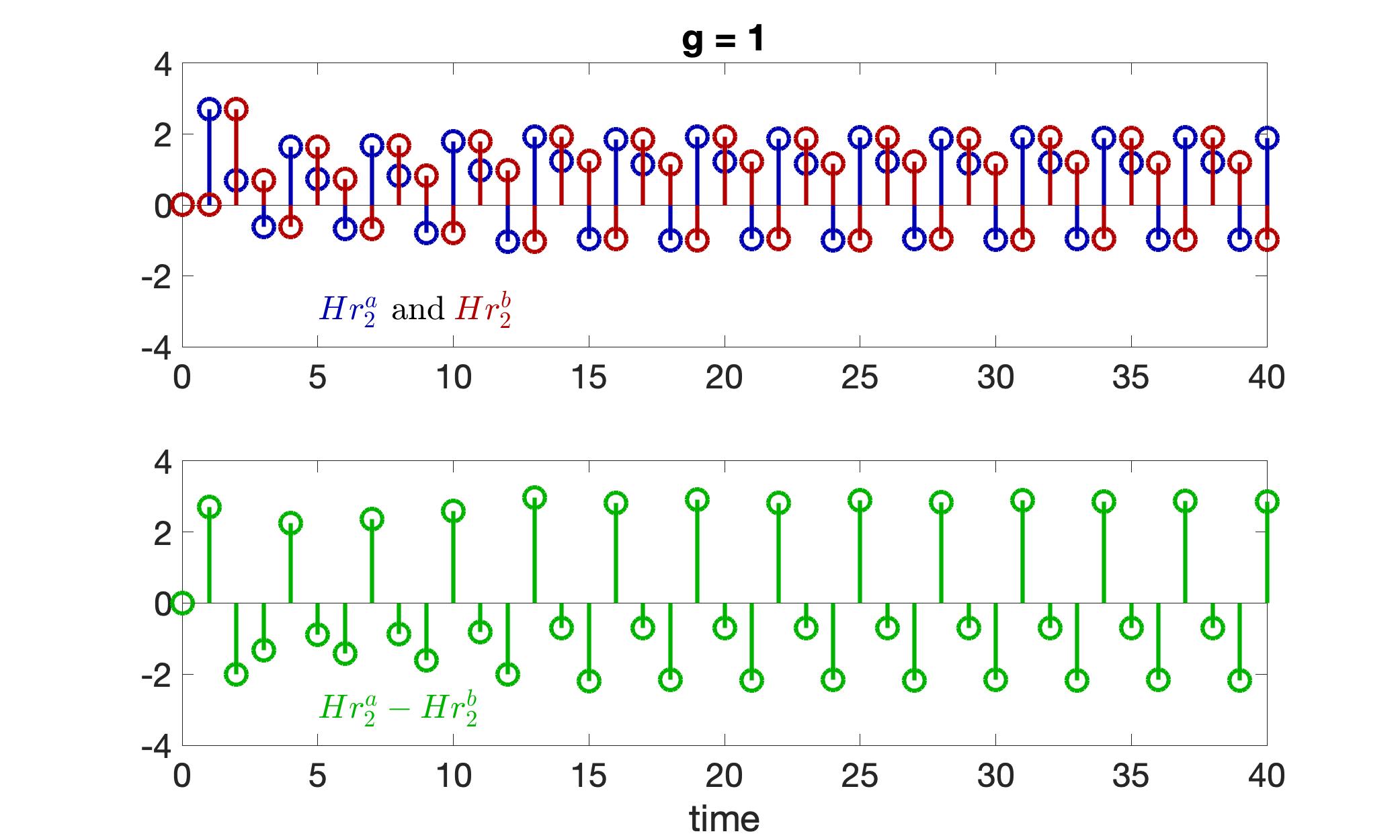}
        \end{center}
        \caption{Respective responses $Hr_2^a$ and $Hr_2^b$ to deterministic step signals $r_2^a$ and $r_2^b$ when the gain is $g=1$. There is no suitable multiplier $M\in\mathcal{M}^d$ for $G$.
        Both responses oscillate: their difference $Hr_2^a-Hr_2^b$ is a power signal but $Hr_2^a-Hr_2^b\notin\ell_2$.}\label{Fig:step4}
    \end{figure}

\subsection{Step response with noise}
    Now suppose $r_2$ is a pulse signal taking values $0$ or $2.7$ with period $400$. In addition $r_1$ is pseudo-random Gaussian white noise with mean $0$ and variance 
    $10^{-3}$.
    Fig.~\ref{fig:noise1} shows the signals $r_1$ and $r_2$ over three periods of $r_2$ while Fig.~\ref{fig:noise2} shows the corresponding signal $u_2$ for the three cases $g=0.6$, $g=0.8$ and $g=1$. While there are quantitative differences between all responses, there is a qualitative difference when $g=1$ and when $r_2$ is high. This is exactly what me might expect, as  the discrete-time counterpart of  Theorem~\ref{thm:power_L2} can be applied to all three cases when $r_2$ is low, but  the discrete-time counterpart of  Theorem~\ref{thm:power_offset} only applies to the cases $g=0.6$ and $g=0.8$.
    For the first two periods with $g=1$ an oscillatory response is evoked when $r_2$ is high, but not for the third period. In this case the response is indeed ``critically sensitive to small changes in inputs'' \cite{Zames66a}.
    
    \begin{figure}[htbp]
        \begin{center}
        \includegraphics[width=\figwidth\textwidth]{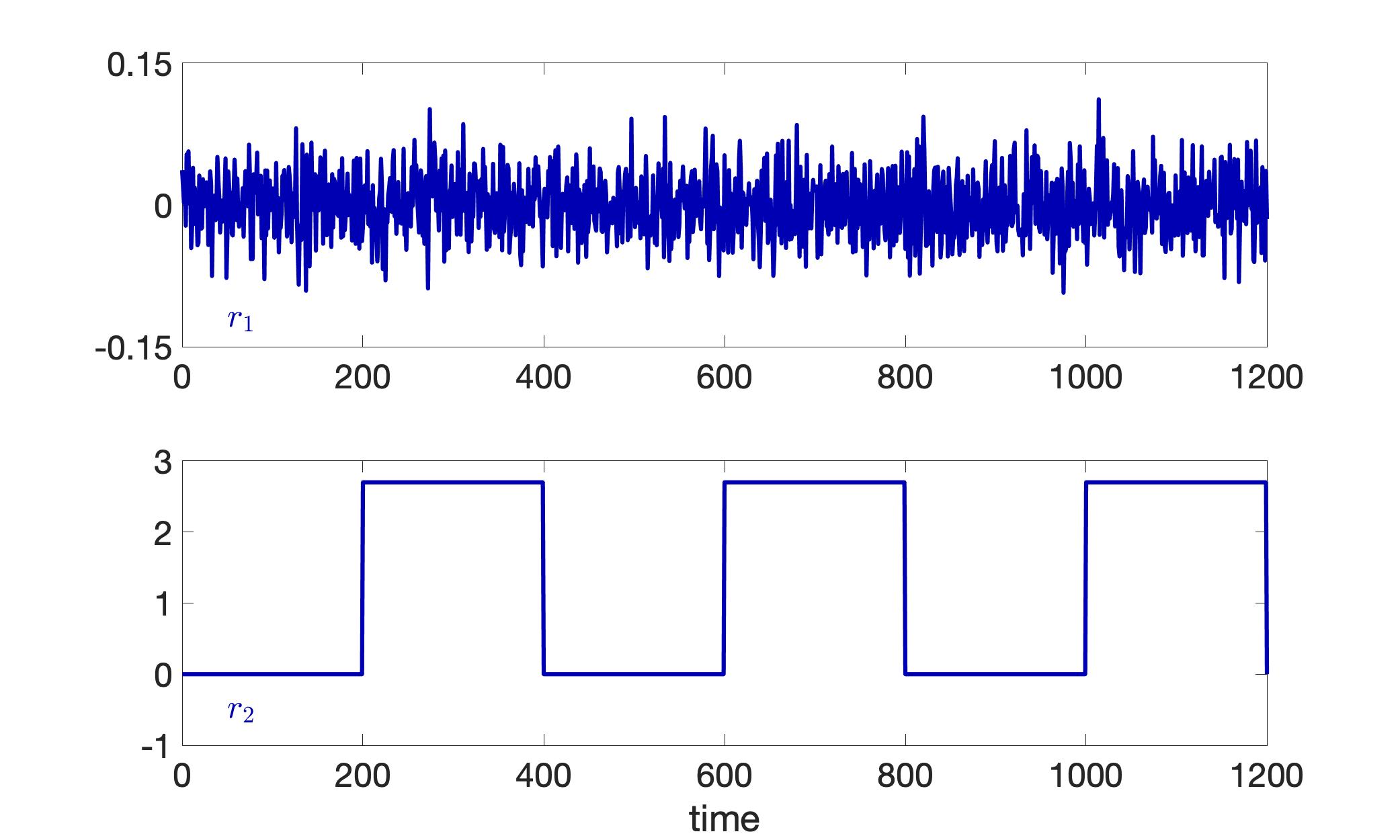}
        \end{center}
        \caption{Exogenous signals $r_1$ and $r_2$ for the step response with noise. Both are power signals.}\label{fig:noise1}
        \begin{center}
        \includegraphics[width=\figwidth\textwidth]{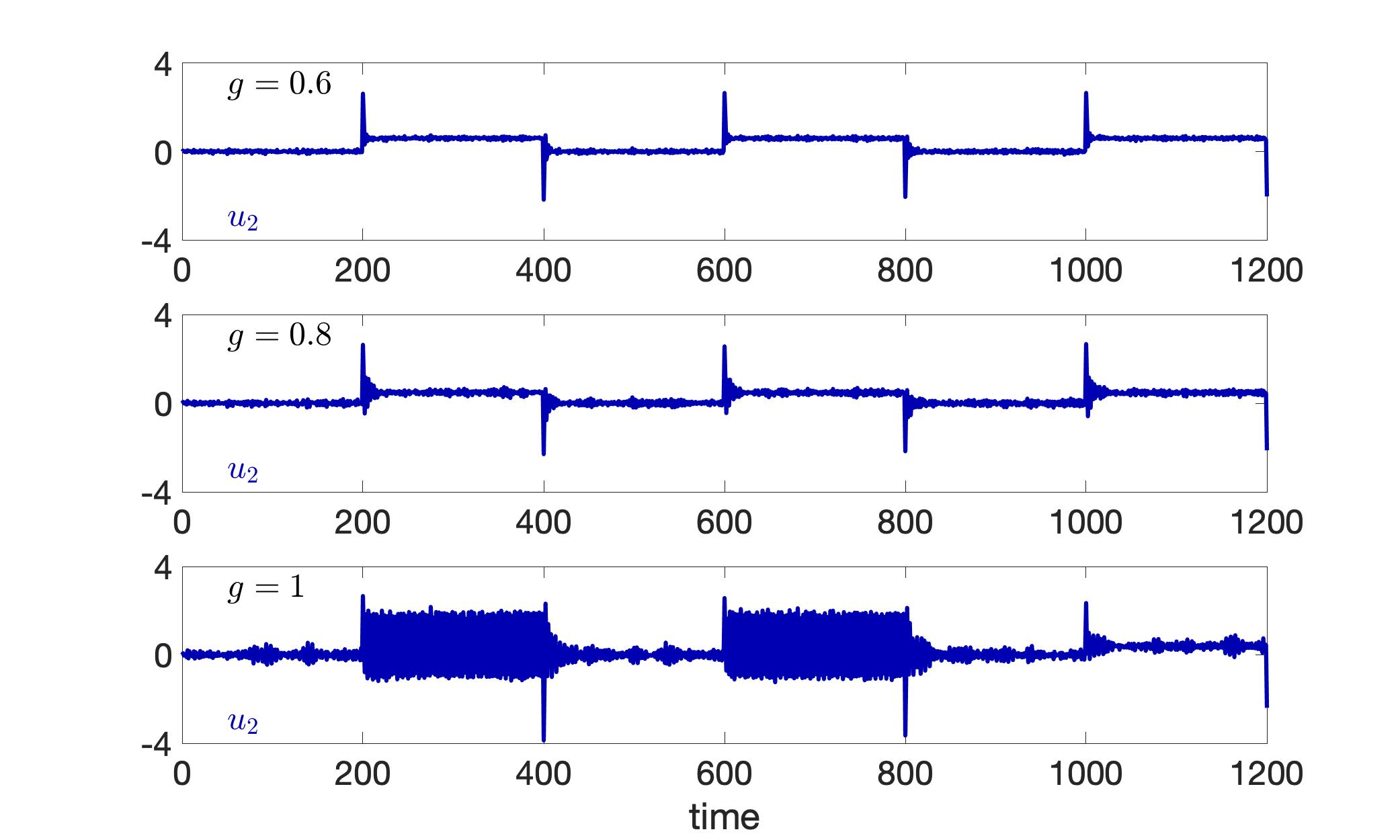}
        \end{center}
        \caption{Response $u_2$ for the cases $g=0.6$, $g=0.8$ and $g=1$. There is a qualitative difference for the case $g=1$ in the time intervals $[200,400]$ and $[600,800]$ (when $r_2$ is high), indicating that the response is  critically sensitive to noise in this case.}\label{fig:noise2}
    \end{figure}


\section{Discussion}

We have shown that the existence of a suitable OZF multiplier in $\mathcal{M}$ (but not $\mathcal{M}_\text{odd}-\mathcal{M}$) ensures a Lurye system is FGOS. This in turn ensures that if the exogenous signal has small power measured around some bias then the output also has small power measured around a uniquely determined bias.

The example of \cite{Fromion04} shows that the result cannot be extended (in general) to exogenous signals that have small power measured around some periodic signal. From a practical point of view it would be interesting to know if there is an extension to signals that are sufficiently slow moving. 

It is worth noting that Theorem~\ref{thm:power_offset} does not exclude discontinuity. The example in \cite{Fromion04} shows Lipschitz discontinuity when the exogenous signal is $r_2(t) = \sin (2 t)$. For $T>0$ sufficiently big the discontinuity will still occur if the input is $r_2(t)=\sin(2t)$ for $t<T$ and $r_2(t)=\sin(2T)$ for $t\geq T$. For this example the presence of small power noise may significantly affect the transient behaviour, even if the output power is guaranteed small.

Finally we note that both the example of \cite{Fromion04} and our example with $g=1$ have high $\mathcal{L}_2$ or $\ell_2$ gain $h$. A lower bound is given by the peak sensitivity when the nonlinearity with a gain corresponding to its maximum slope. For the example in \cite{Fromion04} this lower bound is 
$ \|(1+G)^{-1} \|_{\infty}$ with $G(s) = \frac{909}{(s^2+0.1s+1)(s+100)}$, viz $h\geq 311.35$.
For our example this lower bound is 
$ \|(1+G)^{-1} \|_{\infty}$ with $G$ given by (\ref{defG}) and $g=1$, viz $h\geq 28.58$.
We know of no such example when the gain $h$ is small.






\bibliographystyle{IEEEtran}
\bibliography{ok}

\begin{thebibliography}{10}
\providecommand{\url}[1]{#1}
\csname url@samestyle\endcsname
\providecommand{\newblock}{\relax}
\providecommand{\bibinfo}[2]{#2}
\providecommand{\BIBentrySTDinterwordspacing}{\spaceskip=0pt\relax}
\providecommand{\BIBentryALTinterwordstretchfactor}{4}
\providecommand{\BIBentryALTinterwordspacing}{\spaceskip=\fontdimen2\font plus
\BIBentryALTinterwordstretchfactor\fontdimen3\font minus
  \fontdimen4\font\relax}
\providecommand{\BIBforeignlanguage}[2]{{%
\expandafter\ifx\csname l@#1\endcsname\relax
\typeout{** WARNING: IEEEtran.bst: No hyphenation pattern has been}%
\typeout{** loaded for the language `#1'. Using the pattern for}%
\typeout{** the default language instead.}%
\else
\language=\csname l@#1\endcsname
\fi
#2}}
\providecommand{\BIBdecl}{\relax}
\BIBdecl

\bibitem{Zames68}
G.~Zames and P.~Falb, ``Stability conditions for systems with monotone and
  slope-restricted nonlinearities,'' \emph{SIAM Journal on Control}, vol.~6,
  no.~1, pp. 89--108, 1968.

\bibitem{desoer75}
C.~A. Desoer and M.~Vidyasagar, \emph{Feedback systems: input-output
  properties}.\hskip 1em plus 0.5em minus 0.4em\relax Academic Press, reprinted
  SIAM 2009, 1975.

\bibitem{Veenman16}
J.~Veenman, C.~W. Scherer, and H.~Köroğlu, ``Robust stability and performance
  analysis based on integral quadratic constraints,'' \emph{European Journal of
  Control}, vol.~31, pp. 1--32, 2016.

\bibitem{Megretski04}
A.~Megretski, C.~Kao, U.~Jonsson, and A.~Rantzer, ``A guide to {IQC} $\beta$: A
  {M}atlab toolbox for robust stability and performance analysis,''
  \emph{Technical Report, MIT}, 2004.

\bibitem{Kao04}
C.-Y. Kao, A.~Megretski, U.~Jonsson, and A.~Rantzer, ``A {M}atlab toolbox for
  robustness analysis,'' in \emph{IEEE International Conference on Robotics and
  Automation}, 2004.

\bibitem{Turner12}
M.~C. Turner and M.~L. Kerr, ``Gain bounds for systems with sector bounded and
  slope-restricted nonlinearities,'' \emph{International Journal of Robust and
  Nonlinear Control}, vol.~22, no.~13, pp. 1505--1521, 2012.

\bibitem{Kulkarni02}
V.~Kulkarni and M.~Safonov, ``Incremental positivity nonpreservation by
  stability multipliers,'' \emph{IEEE Transactions on Automatic Control},
  vol.~47, no.~1, pp. 173--177, 2002.

\bibitem{Waitman17}
S.~Waitman, L.~Bako, P.~Massioni, G.~Scorletti, and V.~Fromion, ``Incremental
  stability of {L}ur’e systems through piecewise-affine approximations,'' in
  \emph{20th IFAC World Congress}, 2017.

\bibitem{Zames66a}
G.~Zames, ``On the input-output stability of time-varying nonlinear feedback
  systems. {P}art one: conditions derived using concepts of loop gain,
  conicity, and positivity,'' \emph{IEEE Transactions on Automatic Control},
  vol.~11, no.~2, pp. 228--238, 1966.

\bibitem{Fromion96}
V.~Fromion, S.~Monaco, and D.~Normand-Cyrot, ``Asymptotic properties of
  incrementally stable systems,'' \emph{IEEE Transactions on Automatic
  Control}, vol.~41, no.~5, pp. 721--723, 1996.

\bibitem{Angeli02}
D.~Angeli, ``A {L}yapunov approach to incremental stability properties,''
  \emph{IEEE Transactions on Automatic Control}, vol.~47, no.~3, pp. 410--421,
  2002.

\bibitem{Sepulchre22}
R.~Sepulchre, T.~Chaffey, and F.~Forni, ``On the incremental form of
  dissipativity,'' in \emph{25th International Symposium on Mathematical Theory
  of Networks and Systems MTNS}, 2022.

\bibitem{Chaffey23}
T.~Chaffey, F.~Forni, and R.~Sepulchre, ``Graphical nonlinear system
  analysis,'' \emph{IEEE Transactions on Automatic Control}, vol.~68, no.~10,
  pp. 6067--6081, 2023.

\bibitem{Brockett66a}
R.~Brockett, ``The status of stability theory for deterministic systems,''
  \emph{IEEE Transactions on Automatic Control}, vol.~11, no.~3, pp. 596--606,
  1966.

\bibitem{Benes65}
V.~Beneš and I.~Sandberg, ``On the response of time-variable nonlinear systems
  to almost periodic signals,'' \emph{Journal of Mathematical Analysis and
  Applications}, vol.~10, no.~2, pp. 245--268, 1965.

\bibitem{Fromion04}
V.~Fromion and M.~G. Safonov, ``Popov-{Z}ames-{F}alb multipliers and continuity
  of the input/output map,'' in \emph{6th IFAC Symposium on Nonlinear Control
  Systems (NOLCOS), Stuttgart, Germany}, 2004.

\bibitem{Veenman14}
J.~Veenman and C.~W. Scherer, ``{IQC}-synthesis with general dynamic
  multipliers,'' \emph{International Journal of Robust and Nonlinear Control},
  vol.~24, no.~17, pp. 3027--3056, 2014.

\bibitem{Bertolin22}
A.~L.~J. Bertolin, R.~C. L.~F. Oliveira, G.~Valmorbida, and P.~L.~D. Peres,
  ``Control design of uncertain discrete-time lur'e systems with sector and
  slope bounded nonlinearities,'' \emph{International Journal of Robust and
  Nonlinear Control}, vol.~32, no.~12, pp. 7001--7015, 2022.

\bibitem{vidyasagar93}
M.~Vidyasagar, \emph{Nonlinear systems analysis}.\hskip 1em plus 0.5em minus
  0.4em\relax Prentice-Hall International Editions, reprinted SIAM 2002, 1993.

\bibitem{Zhou96}
K.~Zhou, J.~C. Doyle, and K.~Glover, \emph{Robust and optimal control}.\hskip
  1em plus 0.5em minus 0.4em\relax Prentice-Hall, 1996.

\bibitem{Partington04}
J.~R. Partington, \emph{Linear operators and linear systems: an analytical
  approach to control theory}.\hskip 1em plus 0.5em minus 0.4em\relax CUP,
  2004.

\bibitem{Mari96}
J.~Mari, ``A counterexample in power signals space,'' \emph{IEEE Transactions
  on Automatic Control}, vol.~41, no.~1, pp. 115--116, 1996.

\bibitem{OShea67}
R.~O'Shea, ``An improved frequency time domain stability criterion for
  autonomous continuous systems,'' \emph{IEEE Transactions on Automatic
  Control}, vol.~12, no.~6, pp. 725--731, 1967.

\bibitem{Megretski97}
A.~Megretski and A.~Rantzer, ``System analysis via integral quadratic
  constraints,'' \emph{IEEE Transactions on Automatic Control}, vol.~42, no.~6,
  pp. 819--830, 1997.

\bibitem{Carrasco12}
J.~Carrasco, W.~P. Heath, and A.~Lanzon, ``Factorization of multipliers in
  passivity and {IQC} analysis,'' \emph{Automatica}, vol.~48, no.~5, pp.
  909--916, 2012.

\bibitem{Willems68}
J.~C. Willems and R.~W. Brockett, ``Some new rearrangement inequalities having
  application in stability analysis,'' \emph{IEEE Transactions on Automatic
  Control}, vol.~13, no.~5, pp. 539--549, 1968.

\bibitem{Willems71}
J.~C. Willems, \emph{The analysis of feedback systems}.\hskip 1em plus 0.5em
  minus 0.4em\relax MIT Press, 1971.

\bibitem{Kharitenko23}
A.~Kharitenko and C.~Scherer, ``Time-varying {Z}ames–{F}alb multipliers for
  {LTI} systems are superfluous,'' \emph{Automatica}, vol. 147, p. 110577,
  2023.

\bibitem{Heath15}
W.~P. Heath, J.~Carrasco, and M.~{de la Sen}, ``Second-order counterexamples to
  the discrete-time {K}alman conjecture,'' \emph{Automatica}, vol.~60, pp.
  140--144, 2015.

\bibitem{Heath22}
W.~P. Heath, J.~Carrasco, and D.~A. Altshuller, ``Multipliers for
  nonlinearities with monotone bounds,'' \emph{IEEE Transactions on Automatic
  Control}, vol.~67, no.~2, pp. 910--917, 2022.

\bibitem{Zhang22}
J.~Zhang, J.~Carrasco, and W.~P. Heath, ``Duality bounds for discrete-time
  {Z}ames–{F}alb multipliers,'' \emph{IEEE Transactions on Automatic
  Control}, vol.~67, no.~7, pp. 3521--3528, 2022.

\end{thebibliography}

\end{document}